\documentclass[preprint,amsmath]{revtex4}
\usepackage{amsthm}
\usepackage{float}
\usepackage[dvipsnames, svgnames, x11names]{xcolor}
\usepackage{graphicx} 
\usepackage{epstopdf}
\usepackage{lipsum}
\usepackage{caption}
\usepackage{pifont}
\usepackage{amssymb}
\usepackage{algorithm }
\usepackage{algorithmic}
\usepackage{amsfonts}

\def\qed{\leavevmode\unskip\penalty9999 \hbox{}\nobreak\hfill
     \quad\hbox{\leavevmode  \hbox to.77778em{%
               \hfil\vrule   \vbox to.675em%
               {\hrule width.6em\vfil\hrule}\vrule\hfil}}
     \par\vskip3pt}

\newtheorem{remark}{Remark}

\newtheorem{theorem}{Theorem}
\newtheorem{corollary}{Corollary}
\newtheorem{lemma}{Lemma}
\newtheorem{example}{Example}
\def\be{\begin{eqnarray}}
\def\ee{\end{eqnarray}}
\def\ba{\begin{array}{l}}
\def\ea{\end{array}}

\begin{document}

\title{On monogamy and polygamy relations of multipartite systems}
\author{Xia Zhang$^1$, Naihuan Jing$^{2}$, Ming Liu$^1$, Haitao Ma$^3$\\ 
$^{1}${School of Mathematics, South China University of Technology, Guangzhou 510641, China}\\  
$^{2}${Department of Mathematics, North Carolina State University, Raleigh NC 27695, USA}\\ 
$^{3}${College of Mathematical Science, Harbin Engineering University, Harbin 150001, China} \\
$^\dag$
Corresponding author. E-mail: jing@ncsu.edu}

\begin{abstract}
\baselineskip18pt
We study the monogamy and polygamy relations related to quantum correlations for multipartite quantum systems in a unified manner. It is known that any bipartite measure obeys monogamy and polygamy relations for the $r$-power of the measure. We show in a uniformed manner that the generalized monogamy and polygamy relations are transitive to other powers of the measure in weighted forms.
We demonstrate that our weighted monogamy and polygamy relations are stronger than recently available relations.
Comparisons are given in detailed examples which show that our results are stronger in both situations.

\keywords{Genuine multipartite entanglement \and Correlation tensor \and Weyl operators}
\end{abstract}

\maketitle

\section{Introduction}
Monogamy relations of quantum entanglement are important feature of quantum physics that play an important role in quantum information and quantum communication. Monogamy relations confine the entanglement of a quantum system with the other (sub)systems, thus they are closely related to quantum information processing tasks such as security analysis of quantum key distribution \cite{MP}.	

The monogamy relation for a three-qubit state $\rho_{A_1A_2A_3}$ is defined \cite{MK} as
$$\mathcal{E}(\rho_{A_1|A_2A_3})\geq \mathcal{E}(\rho_{A_1A_2}) +\mathcal{E}(\rho_{A_1A_3}),$$
where $\mathcal{E}$ is a bipartite entanglement measure, $\rho_{A_1A_2}$ and $\rho_{A_1A_3}$ are the reduced density matrices of $\rho_{A_1A_2A_3}$. The monogamy relation was generalized to multiqubit quantum systems, high-dimensional quantum systems in
general settings \cite{ZXN,JZX,jll,012329,gy1,gy2,jin1,jin2, SC, RWF, ZYS}.
The first polygamy relation of entanglement was established in \cite{gg} for some three-qubit system as the inequality ${E_a}_{A_1|A_2A_3}\leq {E_a}_{A_1A_2} +{E_a}_{A_1A_3}$ where ${E_a}_{A_1|A_2A_3}$ is the assisted entanglement \cite{gg} between $A_1$ and $A_2A_3$ and later generalized to some multiqubit systems in \cite{jsb,jin3}. General polygamy inequalities of multipartite entanglement were also given in \cite{062328,295303,jsb,042332} in terms of entanglement of assistance. While it is
known that the monogamy relation does not hold for all quantum systems, it is also shown \cite{GG} that
any monotonic bipartite measure is monogamous on pure tripartite states.

It turns out that a generalized monogamy relation always holds for any quantum system. In \cite{JZX, jll,ZXN},  multiqubit monogamy relations have been demonstrated for the $x$th power of the entanglement of
formation ($x\geq\sqrt{2}$) and the concurrence ($x\geq2$), which opened new direction to study the monogamy relation. Similar general polygamy relation
has been shown for R\'ennyi-$\alpha$ entanglement \cite{GYG}. Monogamy relations for quantum steering have also been shown in \cite{hqy,mko,jk1,jk2,jk3}.
Moreover, polygamy inequalities were given in terms of the
$\alpha$th ($0\leq\alpha\leq 1$) power of square of convex-roof extended negativity (SCREN) and the entanglement of assistance \cite{j012334, 042332}. 

An important feature for this generalized monogamy relation for $\alpha$th power of the measure is its transitivity in the sense that other power of the measure satisfies a weighted monogamous relation. Recently, the authors in \cite{JFQ} provided a class of monogamy and polygamy relations of the $\alpha$th $(0\leq\alpha\leq r,r\geq2)$ and the $\beta$th $(\beta\geq s,0\leq s\leq1)$ powers for any quantum correlation. Applying the monogamy relations in \cite{JFQ} to quantum correlations like squared convex-roof extended negativity, entanglement of formation and concurrence one can get tighter monogamy inequalities than those given in \cite{zhu}. 
Similarly applying the bounds in \cite {JFQ}
 to specific quantum correlations such as the concurrence of assistance, square of convex-roof extended negativity of assistance (SCRENoA), entanglement of assistance,
 corresponding polygamy relations were obtained, which are complementary to the existing ones \cite{jin3,062328,295303,jsb,042332} with different regions of parameter $\beta$. In \cite{ZJZ1,ZJZ2}, the authors gave another set of monogamy relations for  $(0\leq\alpha\leq \frac{r}{2},r\geq2)$ and polygamy relations for $(\beta\geq s,0\leq s\leq1)$, and we note that the bound is stronger than \cite{JFQ} in monogamy case but weaker in polygamy case.

One realizes that the monogamy and polygamy relations given in \cite{JFQ, ZJZ1,ZJZ2} were obtained by bounding the function $(1+t)^x$ by various estimates. In this paper, we revisit the function $(1+t)^x$ and give a unified and improved method to estimate
its upper and lower bounds, which then lead to new improved monogamy and polygamy relations stronger than some of the recent strong ones in both situations.
For instance, we will rigorously show the monogamy and polygamy relations of quantum correlations for the cases $(0\leq\alpha\leq r,r\geq2)$ and $(\beta\geq s,0\leq s\leq1)$ are tighter than those given in \cite{JFQ, ZJZ1,ZJZ2} all together. We also use the concurrence and SCRENoA as examples to demonstrate how our bounds have improved
previously available strong bounds.

\section{Monogamy relations of quantum correlations}

Let $\rho$ be a density matrix on a multipartite quantum system $\bigotimes_{i=1}^nA_i$, and let
$\mathcal{Q}$ be a measure of quantum correlation for any bipartite (sub)system. If $\mathcal{Q}$ satisfies \cite{ARA} the inequality
\begin{eqnarray}\label{q}
&&\mathcal{Q}(\rho_{A_1|A_2,\cdots,A_{n}})\nonumber\\
&&\geq\mathcal{Q}(\rho_{A_1A_2})+\mathcal{Q}(\rho_{A_1A_3})+\cdots+\mathcal{Q}(\rho_{A_1A_{n}}),
\end{eqnarray}
$\mathcal Q$ is said to be {\it monogamous}, where $\rho_{A_1A_i}$, $i=2,...,n$, are the reduced density matrices of $\rho$. For simplicity, we denote $\mathcal{Q}(\rho_{A_1A_i})$ by $\mathcal{Q}_{A_1A_i}$, and $\mathcal{Q}(\rho_{A_1|A_2,\cdots,A_{n}})$ by $\mathcal{Q}_{A_1|A_2,\cdots,A_{n}}$. It is known that some quantum measures obey the monogamous relation \cite{AKE,SSS} for certain quantum states, while
there are quantum measures which do not satisfy the monogamy relation \cite{GLGP, RPAK}.

In Ref. \cite{SPAU,TJ,YKM}, the authors have proved that there exists $r\in \mathbb R~(r\geq2)$ such that the $r$th power of any measure $\mathcal{Q}$ satisfies the following generalized monogamy relation for arbitrary dimensional tripartite state \cite{SPAU}:
\begin{eqnarray}\label{aqq}
\mathcal{Q}^r_{A_1|A_2A_3}\geq\mathcal{Q}^r_{A_1A_2}+\mathcal{Q}^r_{A_1A_3}.
\end{eqnarray}

Assuming \eqref{aqq}, we would like to
prove that other power of $\mathcal Q$ also obeys a weighted monogamy relation. First of all, using the inequality $(1+t)^{x} \geq 1+t^{x}$ for $x \geq 1,0 \leq t \leq 1$ one can easily
derive the following generalized polygamy relation for the $n$-partite case,
$$
\mathcal{Q}_{A_1 \mid A_{2}, \ldots, A_{n}}^{r} \geq \sum_{i=2}^{n} \mathcal{Q}_{A_1 A_{i}}^{r}
$$

We now try to generalize the monogamy relation to other powers. Let's start with a useful lemma.
\begin{lemma}\label{lem:1} Let $a\geq 1$ be a real number.  Then for $t\geq a\geq 1$, the function $(1+t)^x$ satisfies the following inequality
\begin{equation}
(1+t)^x\geq (1+a)^{x-1}+(1+\frac{1}{a})^{x-1}t^x,
\end{equation}
where $0<x\leq 1$.
\end{lemma}
\begin{proof}.
Consider $g(x,y)=(1+y)^x-(1+{a})^{x-1}y^x$, $0<y\leq \frac{1}{a}$, $0\leq x\leq 1$. Then
$$\frac{\partial g}{\partial y}=xy^{x-1}\left((1+\frac{1}{y})^{x-1}-(1+a)^{x-1}\right).$$
Let $h(x,y)=(1+\frac{1}{y})^{x-1}$, $0<y\leq \frac{1}{a}$, $0\leq x\leq 1$. Since $h(x,y)$ is an increasing function of $y$, we have $h(x,y)\leq h(x,\frac{1}{a})=(1+{a})^{x-1}$.
Thus $\frac{\partial g}{\partial y}\leq 0$ and $g(x,y)$ is decreasing with respect to $y$. Therefore we have $g(x,y)\geq g(x,\frac{1}{a})=(1+\frac{1}{a})^{x-1}$. Subsequently
\begin{equation*}
g(x,\frac{1}{t})=\frac{(1+t)^x}{t^x}-\frac{(1+{a})^{x-1}}{t^x}\geq (1+\frac{1}{a})^{x-1}
\end{equation*}
for $t\geq a$. Then we have
\begin{equation*}
(1+t)^x\geq(1+{a})^{x-1}+(1+\frac{1}{a})^{x-1}t^x,
\end{equation*}
for $t\geq a$.

\end{proof}
\begin{remark}\label{rem:1}~~
The proof of Lemma \ref{lem:1} implies that for $0\leq x\leq 1$, $t\geq a\geq 1$ and $f(x)\geq (1+{a})^{x-1} $, we have
\begin{equation*}
(1+t)^x\geq f(x)+\frac{(1+a)^x-f(x)}{a^x}t^x.
\end{equation*}

It is not hard to see that for $0\leq x\leq 1$, $t\geq a\geq 1$ and for $f(x)\geq (1+{a})^{x-1}$,
$$(1+{a})^{x-1}+(1+\frac{1}{a})^{x-1}t^x-
\left[f(x)+\frac{(1+a)^x-f(x)}{a^x}t^x\right]\geq 0.$$
In fact, 
\begin{equation*}
\begin{aligned}
{ LHS} 
&=(1+{a})^{x-1}-f(x)+\Big(\frac{a(1+a)^{x-1}}{a^x}t^{x}-\frac{(1+a)(1+a)^{x-1}-f(x)}{a^x}t^{x}\Big)\\
&=(1+{a})^{x-1}-f(x)+\frac{f(x)-(1+{a})^{x-1}}{a^x}t^{x}\\
&=(\frac{t^x}{a^x}-1)(f(x)-(1+{a})^{x-1})\geq 0.
\end{aligned}
\end{equation*}
\end{remark}
\begin{remark} \label{rem:2}
In \cite[Lemma 1]{JFQ}, the authors have given a lower bound of $(1+t)^x$ for $0\leq x\leq 1$ and $t\geq a\geq 1$:
\begin{equation*}
(1+t)^{x} \geq 1+\frac{(1+a)^{x}-1}{a^{x}} t^{x}.
\end{equation*}
This is a special case of our result in Remark \ref{rem:1}. In fact, let $f(x)=1\geq (1+{a})^{x-1}$ for $0\leq x\leq 1$ in Remark \ref{rem:1}, then the inequality
descends to theirs. Therefore our lower bound of $(1+t)^x$ is better than that of \cite{JFQ}, consequently
any monogamy relations based on Lemma \ref{lem:1} are better than those given in \cite{JFQ} based on Lemma 1 of \cite{JFQ}.
\end{remark}

\begin{remark}\label{rem:ZJZ1}
In \cite[Lemma 1]{ZJZ1}, the authors gave another lower bound of $(1+t)^x$ for $0\leq x\leq \frac{1}{2}$ and $t\geq a\geq 1$
\begin{equation*}
(1+t)^{x} \geq p^{x}+\frac{(1+a)^{x}-p^{x}}{a^{x}} t^{x},
\end{equation*}
where $\frac{1}{2}\leq p\leq 1$.
This is a also special case of our Remark \ref{rem:1} where $f(x)=p^{x}$ for $0\leq x\leq \frac{1}{2}$ and $t\geq a\geq 1$.
Since $(1+a)^{x-1}\leq p^x$ for $0\leq x\leq \frac{1}{2}$ and $\frac{1}{2}\leq p\leq 1$,
therefore our lower bound of $(1+t)^x$ for $0\leq x\leq \frac{1}{2}$ and $t\geq a\geq 1$ is stronger than that given in \cite{ZJZ1}.
Naturally our monogamy relations based on Lemma \ref{lem:1} will outperform those given in \cite{ZJZ1} based on Lemma 1 of \cite{ZJZ1}.
\end{remark}
\begin{remark}\label{rem:ZJZ2}
In \cite[Lemma 1]{ZJZ2}, the following lower bound was given: $(1+t)^x$ for $0\leq x\leq \frac{1}{2}$ and $t\geq a\geq 1$
\begin{equation*}
(1+t)^{x} \geq (\frac{1}{2})^{x}+\frac{(1+a)^{x}-(\frac{1}{2})^{x}}{a^{x}} t^{x},
\end{equation*}
which is the special case  $p=\frac{1}{2}$ of \cite{ZJZ1}.
therefore our lower bound of $(1+t)^x$ for $0\leq x\leq \frac{1}{2}$ and $t\geq a\geq 1$ is better than the one given in \cite{ZJZ2}.
Thus, our monodamy relations based on Lemma \ref{lem:1} are better than the ones given in \cite{ZJZ2} based on  \cite[Lemma 1]{ZJZ2}.
\end{remark}

\begin{lemma}\label{lem:2}
Let $p_i$ $(i=1,\cdots, n)$ be nonnegative numbers arranged as $p_{(1)}\geq p_{(2)}\geq ...\geq p_{(n)}$
for a permutation $(1)(2)\cdots (n)$ of $12\cdots n$.
If $p_{(i)}\geq a p_{(i+1)}$ for $i=1,...,n-1$,
we have
\begin{equation}\label{eq:3}
\left(\sum_{i=1}^n p_i\right)^x\geq (1+{a})^{x-1}\sum_{i=1}^n \left((1+\frac{1}{a})^{x-1}\right)^{n-i}p_{(i)}^x,
\end{equation}
for $0\leq x\leq 1$.

\end{lemma}

\begin{proof}
For $0\leq x\leq 1$, if $p_{(n)}>0$ using Lemma \ref{lem:1} we have
\begin{equation*}
\begin{aligned}
\left(\sum_{i=1}^n p_i\right)^x&=p_{(n)}^{x}\left(1+\frac{p_{(1)}+...+p_{(n-1)}}{p_{(n)}}\right)^x\\
&\geq (1+{a})^{x-1}p_{(n)}^x+(1+\frac{1}{a})^{x-1}(p_{(1)}+...+p_{(n-1)})^x\\
&\geq...\\
&\geq (1+{a})^{x-1}\sum_{i=1}^n \left((1+\frac{1}{a})^{x-1}\right)^{n-i}p_{(i)}^x.
\end{aligned}
\end{equation*}
The other cases can be easily checked by Lemma \ref{lem:1}. 
\end{proof}

\begin{theorem}\label{thm:1}
For any tripartite mixed state $\rho_{A_1A_2A_3}$, let $\mathcal{Q}$ be a (bipartite) quantum measure satisfying the generalized monogamy relation \eqref{aqq} for $r\geq 2$
and $\mathcal{Q}_{A_1A_3}^{r}\geq a\mathcal{Q}_{A_1A_2}^{r}$ for some $a$, then
we have
\begin{equation}
\mathcal{Q}^{\alpha}_{A_1|A_2A_3}\geq (1+{a})^{\frac{\alpha}{r}-1}\mathcal{Q}_{A_1A_2}^{\alpha}+(1+\frac{1}{a})^{\frac{\alpha}{r}-1}\mathcal{Q}_{A_1A_3}^{\alpha}.
\end{equation}
for $0\leq \alpha\leq r$. 
\end{theorem}
\begin{proof} It follows from Lemma \ref{lem:1} that
\begin{equation}
\begin{aligned}
 \mathcal{Q}_{A_1|A_2A_3}^{\alpha}&=\left(\mathcal{Q}^{r}_{A_1|A_2A_3}\right)^{\frac{\alpha}{r}}\geq \left(\mathcal{Q}_{A_1A_2}^{r}+\mathcal Q_{A_1A_3}^{r}\right)^{\frac{\alpha}{r}}\\
&=\mathcal{Q}_{A_1A_2}^{\alpha}\left(1+\frac{\mathcal{Q}_{A_1A_3}^{r}}{\mathcal{Q}_{A_1A_2}^{r}}\right)^{\frac{\alpha}{r}}\\
&\geq(1+{a})^{\frac{\alpha}{r}-1}\mathcal{Q}_{A_1A_2}^{\alpha}+(1+\frac{1}{a})^{\frac{\alpha}{r}-1}\mathcal{Q}_{A_1A_3}^{\alpha}.
\end{aligned}
\end{equation}
\end{proof}
\begin{theorem}\label{thm:2}
For any n-partite quantum state $\rho_{A_1A_2...A_n}$, let $\mathcal{Q}$ be a (bipartite) quantum measure satisfying the generalized monogamy relation \eqref{aqq} for $r\geq 2$.
Arrange $\mathcal{Q}_{(1)}\geq \mathcal{Q}_{(2)}\geq...\geq \mathcal{Q}_{(n-1)}$ with $\mathcal{Q}_{(j)}\in\{\mathcal{Q}_{A_1A_i}|i=2,...,n\}, j=1,...,n-1$. If for some $a$, $\mathcal{Q}^r_{(i)}\geq a \mathcal{Q}^r_{(i+1)}$ for $i=1,...,n-2$, then we have
\begin{equation}
\mathcal{Q}^{\alpha}_{A_1|A_2...A_n}\geq(1+{a})^{\frac{\alpha}{r}-1}\sum_{i=1}^{n-1} \left((1+\frac{1}{a})^{\frac{\alpha}{r}-1}\right)^{n-1-i}\mathcal{Q}_{(i)}^{{\alpha}}
\end{equation}
for $0\leq \alpha\leq r$. 
\end{theorem}

\begin{proof}
By Lemma \ref{lem:2}, we have
\begin{equation*}
\begin{aligned}
\mathcal{Q}^{\alpha}_{A_1|A_2...A_n}&=\left(\mathcal{Q}^{r}_{A_1|A_2...A_n}\right)^{\frac{\alpha}{r}}
\geq \left(\mathcal{Q}_{A_1A_2}^{r}+...+\mathcal{Q}_{A_1A_n}^{r}\right)^{\frac{\alpha}{r}}\\
&\geq(1+{a})^{\frac{\alpha}{r}-1}\sum_{i=1}^{n-1} \left((1+\frac{1}{a})^{\frac{\alpha}{r}-1}\right)^{n-1-i}\mathcal{Q}_{(i)}^{{\alpha}},
\end{aligned}
\end{equation*}
\end{proof}

The general monogamy relations work for any quantum correlation measure such as concurrence, negativity, entanglement of formation etc.
Thus our theorems produce
tighter weighted monogamy relations than the existing ones (cf. \cite{JFQ,zhu,ZJZ1,ZJZ2}). Moreover, the new weighted monogamy relations also can be used for
the Tsallis-$q$ entanglement and R\'enyi-$q$ entanglement measures, and they also
 outperform some of the recently found monogamy relations in \cite{jll,jin3,slh}. 

In the following, we use the concurrence as an example to show the advantage of our monogamy relations.

For a bipartite pure state $\rho=|\psi\rangle_{AB}\in{H}_A\otimes {H}_B$, the concurrence is defined \cite{AU,PR,SA} by $C(|\psi\rangle_{AB})=\sqrt{{2\left[1-\mathrm{Tr}(\rho_A^2)\right]}}$,
where $\rho_A$ is the reduced density matrix. For a mixed state $\rho_{AB}$ the concurrence is given by the convex roof extension
$C(\rho_{AB})=\min_{\{p_i,|\psi_i\rangle\}}\sum_ip_iC(|\psi_i\rangle)$,
where the minimum is taken over all possible decompositions of $\rho_{AB}=\sum\limits_{i}p_i|\psi_i\rangle\langle\psi_i|$, with $p_i\geq0$, $\sum\limits_{i}p_i=1$ and $|\psi_i\rangle\in {H}_A\otimes {H}_B$.

For convenience, we write $C_{A_1A_i}=C(\rho_{A_1A_i})$ and $C_{A_1|A_2,\cdots,A_{n}}=C(\rho_{A_1|A_2\cdots A_{n}})$. The following conclusions are easily seen by
the similar method as in the proof of Theorem 1.

\begin{corollary}  Let $C$ be the concurrence satisfying the generalized monogamy relation \eqref{aqq} for $r\geq 2$.
For any 3-qubit mixed state $\rho_{A_1A_2A_3}$, if $C_{A_1A_3}^{r}\geq aC_{A_1A_2}^{r}$ for some $a$, then
we have
\begin{equation}
C^{\alpha}_{A_1|A_2A_3}\geq (1+{a})^{\frac{\alpha}{r}-1}C_{A_1A_2}^{\alpha}+(1+\frac{1}{a})^{\frac{\alpha}{r}-1}C_{A_1A_3}^{\alpha}.
\end{equation}
for $0\leq \alpha\leq r$. 
\end{corollary}
\begin{corollary}  Let $C$ be the concurrence satisfying the generalized monogamy relation \eqref{aqq} for $r\geq 2$.
For any n-qubit quantum state $\rho_{A_1A_2...A_n}$, arrange
$C_{(1)}\geq C_{(2)}\geq...\geq C_{(n-1)}$ with $C_{(j)}\in\{C_{A_1A_i}|i=2,...,n\}, j=1,...,n-1$. If for some $a$, $C^r_{(i)}\geq a C^r_{(i+1)}$ for $i=1,...,n-2$, then we have
\begin{equation}
C^{\alpha}_{A_1|A_2...A_n}\geq(1+{a})^{\frac{\alpha}{r}-1}\sum_{i=1}^{n-1} \left((1+\frac{1}{a})^{\frac{\alpha}{r}-1}\right)^{n-1-i}C_{(i)}^{{\alpha}}
\end{equation}
for $0\leq \alpha\leq r$. 
\end{corollary}

\begin{example}
 Consider the following three-qubit state $|\psi\rangle$ with generalized Schmidt decomposition \cite{AA,XH},
$$
|\psi\rangle=\lambda_{0}|000\rangle+\lambda_{1} e^{i \varphi}|100\rangle+\lambda_{2}|101\rangle+\lambda_{3}|110\rangle+\lambda_{4}|111\rangle,
$$
where $\lambda_{i} \geq 0$ and $\sum_{i=0}^{4} \lambda_{i}^{2}=1$. Then $C_{A_1 \mid A_2 A_3}=2 \lambda_{0} \sqrt{\lambda_{2}^{2}+\lambda_{3}^{2}+\lambda_{4}^{2}}, C_{A_1 A_2}=$ $2 \lambda_{0} \lambda_{2}$, and $C_{A_1 A_3}=2 \lambda_{0} \lambda_{3}$.

Set $\lambda_{0}=\lambda_3=\frac{1}{2}, \lambda_{1}=\lambda_{2}=\lambda_{4}=\frac{\sqrt{6}}{6}$.
We have
$C_{A_1 \mid A_2A_3}=\frac{\sqrt{21}}{6}, C_{A_1 A_2}=\frac{\sqrt{6}}{6}, C_{A_1A_3}=\frac{1}{2}$. Set $a=\sqrt{6}/2$, then
the lower bound of $C_{A_1 \mid A_2A_3}^{\alpha}$ given in \cite{JFQ} is
 \begin{equation*}
C_{A_1A_2}^{\alpha}+\frac{(1+a)^{\frac{\alpha}{r}}-1}{a^{\frac{\alpha}{r}}} C_{A_1 A_3}^{\alpha}=\left(\frac{\sqrt{6}}{6}\right)^{\alpha}+\frac{(1+\frac{\sqrt{6}}{2})^{\frac{\alpha}{r}}-1}{(\frac{\sqrt{6}}{2})^{\frac{\alpha}{r}}}\left(\frac{1}{2}\right)^{\alpha}=Z_1,
\end{equation*}
the lower bound of $C_{A_1 \mid A_2A_3}^{\alpha}$ given in \cite{ZJZ1,ZJZ2} is
 \begin{equation*}
C_{A_1A_2}^{\alpha}p^{\frac{\alpha}{r}}+\frac{(1+a)^{\frac{\alpha}{r}}-p^{\frac{\alpha}{r}}}{a^{\frac{\alpha}{r}}} C_{A_1 A_3}^{\alpha}=\left(\frac{\sqrt{6}}{6}\right)^{\alpha}(\frac{1}{2})^{\frac{\alpha}{r}}+\frac{(1+\frac{\sqrt{6}}{2})^{\frac{\alpha}{r}}-(\frac{1}{2})^{\frac{\alpha}{r}}}{(\frac{\sqrt{6}}{2})^{\frac{\alpha}{r}}}\left(\frac{1}{2}\right)^{\alpha}=Z_2,
\end{equation*}
with $p=\frac{1}{2}$ and $\alpha\leq \frac{r}{2}$.
While our bound is
\begin{equation*}
(1+{a})^{\frac{\alpha}{r}-1}C_{A_1A_2}^{\alpha}+(1+\frac{1}{a})^{\frac{\alpha}{r}-1}C_{A_1A_3}^{\alpha}=(1+\frac{\sqrt{6}}{2})^{\frac{\alpha}{r}-1}\left(\frac{\sqrt{6}}{6}\right)^{\alpha}+(1+\frac{2}{\sqrt{6}})^{\frac{\alpha}{r}-1}\left(\frac{1}{2}\right)^{\alpha}=Z_3
\end{equation*}
Fig. 1 charts the graphs of the three bounds and the figure clearly shows that our result is better than those in \cite{JFQ, ZJZ1, ZJZ2} for $0 \leq \alpha \leq 1$ and $r \geq 2$.

\begin{figure}[!htb]
\centerline{\includegraphics[width=0.6\textwidth]{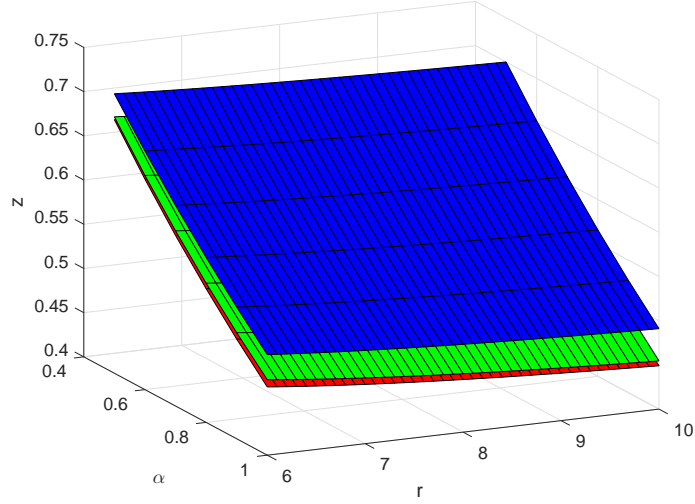}}
\renewcommand{\figurename}{Fig.}
\caption{The $z$-axis shows the concurrence as a function of $\alpha, r$. The blue, green and red surfaces represent our lower bound, the lower bound from \cite{ZJZ1,ZJZ2} and the lower bound from \cite{JFQ} respectively.}
\end{figure}

\end{example}

\section{Polygamy relations for general quantum correlations}
In  \cite{jinzx}, the authors proved that for arbitrary dimensional tripartite states, there exists $0 \leq s \leq 1$ such that any quantum correlation measure $\mathcal{Q}$ satisfies the following polygamy relation:
\begin{equation}\label{poly}
\mathcal{Q}_{A \mid B C}^{s} \leq \mathcal{Q}_{A B}^{s}+\mathcal{Q}_{A C}^{s}.
\end{equation}
Using the similar method as Lemma \ref{lem:1} and Lemma \ref{lem:2}, we can prove the following Lemma \ref{lem:4} and Lemma \ref{lem:5}.
\begin{lemma}\label{lem:4}
For  $t\geq a\geq 1$, $x\geq 1$ we have,
\begin{equation}
(1+t)^x\leq (1+{a})^{x-1}+(1+\frac{1}{a})^{x-1}t^x.
\end{equation}

\end{lemma}

\begin{lemma}\label{lem:5}
For nonnegative numbers $p_i$, $i=1,\cdots, n$, rearrange them in descending order: $p_{(1)}\geq p_{(2)}\geq ...\geq p_{(n)}$
where $p_{(i)}\in \{p_j|j=1,\cdots, n\}$. If $p_{(i)}\geq a p_{(i+1)}$ for $i=1,...,n-1$ and $a$, then
we have
\begin{equation}
\left(\sum_{i=1}^n p_i\right)^x\leq (1+{a})^{x-1}\sum_{i=1}^n \left((1+\frac{1}{a})^{x-1}\right)^{n-i}p_{(i)}^x,
\end{equation}
for $x\geq 1$.
\end{lemma}
\begin{remark}\label{rem:3}~~
Similar argument as Remark \ref{rem:1} implies that
for the case $x\geq 1$, $t\geq a\geq 1$ and $f(x)\leq (1+{a})^{x-1} $, we have
\begin{equation*}
(1+t)^x\leq f(x)+\frac{(1+a)^x-f(x)}{a^x}t^x.
\end{equation*}

We can easily check that for $ x\geq 1$, $t\geq a\geq 1$ and $f(x)\leq (1+{a})^{x-1}$,
$$(1+{a})^{x-1}+(1+\frac{1}{a})^{x-1}t^x-
\left[f(x)+\frac{(1+a)^x-f(x)}{a^x}t^x\right]\leq 0.$$

\end{remark}
\begin{remark}\label{rem:4}
In Lemma 2 of \cite{JFQ}, the authors gave an upper bound of $(1+t)^x$ for $x\geq 1$ and $t\geq a\geq 1$
\begin{equation*}
(1+t)^{x} \leq 1+\frac{(1+a)^{x}-1}{a^{x}} t^{x}.
\end{equation*}
Actually this is a special case of $f(x)=1\leq (1+{a})^{x-1}$ for $ x\geq 1$ in Remark \ref{rem:3},
therefore our upper bound of $(1+t)^x$ is better than the one given in \cite{JFQ}.
Consequently our polygamy relations based on Lemma \ref{lem:4} are better than those given in \cite{JFQ} based on Lemma 2 of \cite{JFQ}.

\end{remark}
\begin{remark}\label{rem:Jing's paper1}
In \cite[Lemma 2]{ZJZ1}, the authors also gave an upper bound of $(1+t)^x$ for $x\geq 1$, $t\geq a\geq 1$ and $0<q\leq 1$
\begin{equation*}
(1+t)^{x} \leq q^{x}+\frac{(1+a)^{x}-q^{x}}{a^{x}} t^{x},
\end{equation*}
Actually this is the special cases of $f(x)=q^x\leq (1+{a})^{x-1}$ for $ x\geq 1$ in Remark \ref{rem:3},
therefore our upper bound of $(1+t)^x$ is better than the one given in \cite{ZJZ1}.
Thus, our polygamy relations based on Lemma \ref{lem:4} are better than those given in \cite{ZJZ1} based on Lemma 2 of \cite{ZJZ1}.
\end{remark}
\begin{remark}\label{rem:Jing's paper2}
In \cite[Lemma 1]{ZJZ2}, an upper bound of $(1+t)^x$ for $x\geq 1$, $t\geq a\geq 1$ was given:
\begin{equation*}
(1+t)^{x} \leq (\frac{1}{2})^{x}+\frac{(1+a)^{x}-(\frac{1}{2})^{x}}{a^{x}} t^{x},
\end{equation*}
Again this is a special case of $f(x)=(\frac{1}{2})^{x}\leq (1+{a})^{x-1}$ for $ x\geq 1$ in Remark \ref{rem:3},
therefore our upper bound of $(1+t)^x$ is better than the one given in \cite{ZJZ2}. Naturally
our polygamy relations based on Lemma \ref{lem:4} are stronger than those of \cite{ZJZ2} based on Lemma 1 of \cite{ZJZ2}.
\end{remark}
\begin{theorem} Let $\mathcal{Q}$ be a bipartite measure satisfying the generalized polygamy relation \eqref{poly} for $0 \leq s \leq 1$.
Suppose $\mathcal{Q}_{A_1A_3}^{s} \geq a \mathcal{Q}_{A_1 A_2}^{s}$
for $a\geq 1$ on any tripartite state $\rho_{A B C} \in$ $H_{A_1} \otimes H_{A_2} \otimes H_{A_3}$, then the quantum correlation measure $\mathcal{Q}$ satisfies
$$
\mathcal{Q}_{A_1 \mid A_2 A_3}^{\beta} \leq (1+{a})^{\frac{\beta}{s}-1}\mathcal{Q}_{A_1 A_2}^{\beta}+(1+\frac{1}{a})^{\frac{\beta}{s}-1} \mathcal{Q}_{A_1 A_3}^{\beta}
$$
for $\beta \geq s$.
\end{theorem}
\begin{theorem} Let $\rho$ be a state on the multipartite system $A_1A_2...A_n$. Let $\mathcal{Q}$ be a bipartite measure satisfying the generalized polygamy relation \eqref{poly} for $0 \leq s \leq 1$. Set $\mathcal{Q}_{(1)}\geq \mathcal{Q}_{(2)}\geq...\geq \mathcal{Q}_{(n-1)}$ with $\mathcal{Q}_{(j)}\in\{\mathcal{Q}_{A_1A_i}|i=2,...,n\}, j=1,...,n-1$. If
$\mathcal{Q}_{(i)}^{s}\geq a \mathcal{Q}_{(i+1)}^s$ for $a$ and $i=1,...,n-2$, then we have
\begin{equation}
\mathcal{Q}^{\beta}_{A_1|A_2...A_n}\leq(1+{a})^{\frac{\beta}{s}-1}\sum_{i=1}^{n-1} \left((1+\frac{1}{a})^{\frac{\beta}{s}-1}\right)^{n-1-i}\mathcal{Q}_{(i)}^{{\beta}}
\end{equation}
for $\beta\geq s$. 
\end{theorem}
\begin{proof}
Since $\mathcal{Q}_{A \mid B C}^{s} \leq \mathcal{Q}_{A B}^{s}+\mathcal{Q}_{A C}^{s}$, we have
\begin{equation}
\mathcal{Q}^{s}_{A_1|A_2...A_n}\leq \mathcal{Q}^{s}_{A_1|A_2}+\mathcal{Q}^{s}_{A_1|A_3...A_n}\leq\cdots\leq \sum_{i=2}^n\mathcal{Q}^{s}_{A_1|A_i}=\sum_{j=1}^{n-1}\mathcal{Q}^{s}_{(j)}.
\end{equation}
By Lemma \ref{lem:5} we have
\begin{equation}
\begin{aligned}
\mathcal{Q}^{\beta}_{A_1|A_2...A_n}&=\left(\mathcal{Q}^{s}_{A_1|A_2...A_n}\right)^{\frac{\beta}{s}}\leq \left(\sum_{j=1}^{n-1}\mathcal{Q}^{s}_{(j)}\right)^{\frac{\beta}{s}}\\
&\leq(1+{a})^{\frac{\beta}{s}-1}\sum_{i=1}^{n-1} \left((1+\frac{1}{a})^{\frac{\beta}{s}-1}\right)^{n-1-i}\mathcal{Q}_{(i)}^{{\beta}}
\end{aligned}
\end{equation}
\end{proof}

The general monogamy relations can be applied to any quantum correlation measure such as the concurrence of assistance, square of convex-roof extended negativity of assistance (SCRENoA), and entanglement of assistance etc. Correspondingly new class of (weighted) polygamy relations are obtained. In the following, we take the SCRENoA as an example.

The negativity of bipartite state $\rho_{A_1A_2}$ is defined by \cite{GRF}:
$N(\rho_{A_1A_2})=(||\rho_{A_1A_2}^{T_{A_{1}}}||-1)/2$,
where $\rho_{A_1A_2}^{T_{A_1}}$ is the partial transposition with respect to the subsystem $A_1$ and $||X||=\mathrm{Tr}\sqrt{XX^\dag}$ is the trace norm of $X$.
 For convennience, we use the following definition of negativity, $ N(\rho_{A_1A_2})=||\rho_{A_1A_2}^{T_{A_{1}}}||-1$.
For any bipartite pure state $|\psi\rangle_{A_1A_2}$, the negativity $ N(\rho_{A_1A_2})$ is given by
$$N(|\psi\rangle_{A_1A_2})=2\sum_{i<j}\sqrt{\lambda_i\lambda_j}=(\mathrm{Tr}\sqrt{\rho_{A_1}})^2-1,$$
where $\lambda_i$ are the eigenvalues for the reduced density matrix $\rho_A$ of $|\psi\rangle_{A_1A_2}$. For a mixed state $\rho_{A_1A_2}$, the square of convex-roof extended negativity (SCREN) is defined by
 $N_{sc}(\rho_{A_1A_2})=[\mathrm{min}\sum_ip_iN(|\psi_i\rangle_{A_1A_2})]^2$,
where the minimum is taken over all possible pure state decompositions $\{p_i,~|\psi_i\rangle_{A_1A_2}\}$ of $\rho_{A_1A_2}$. The SCRENoA is then defined by $N_{sc}^a(\rho_{A_1A_2})=[\mathrm{max}\sum_ip_iN(|\psi_i\rangle_{A_1A_2})]^2$, where the maximum is taken over all possible pure state decompositions $\{p_i,~|\psi_i\rangle_{A_1A_2}\}$ of $\rho_{A_1A_2}$. For convenience, we denote ${N_a}_{A_1A_i}=N_{sc}^a(\rho_{A_1A_i})$ the SCRENoA of $\rho_{A_1A_i}$ and ${N_a}_{A_1|A_2\cdots A_{n}}=N^a_{sc}(|\psi\rangle_{A_1|A_2\cdots A_{n}})$.

\begin{corollary}  Let $s\in (0, 1)$ be the fixed number so that the SCRENoA satisfying the generalized polygamy relation \eqref{poly}.
Suppose ${N_a}^s_{A_1A_3}\geq a{N_a}^s_{A_1A_2}$
for $a\geq 1$ on a $2\otimes2\otimes2^{N-2}$ tripartite mixed state $\rho$, then the SCRENoA satisfies
\begin{eqnarray}\label{co31}
{N_a}^\beta_{A_1|A_2A_3}\leq(1+{a})^{\frac{\beta}{s}-1}{N_a}^\beta_{A_1A_2}+(1+\frac{1}{a})^{\frac{\beta}{s}-1} {N_a}^\beta_{A_1A_3}
\end{eqnarray}
for $\beta\geq \delta$.
\end{corollary}

By induction, the following result is immediate for a multiqubit quantum state $\rho_{A_1A_2\cdots A_{n}}$. 
\begin{corollary}  Let $s\in (0, 1)$ be the fixed number so that the SCRENoA satisfying the generalized polygamy relation \eqref{poly}.
Let $N_{a(1)}\geq N_{a(2)}\geq...\geq N_{a(n-1)}$ be a reordering of $N_{aA_1A_j}$, $j=2,...,n$.
If $N_{a(i)}^{s}\geq a N_{a(i+1)}^s$ for $a$ and $i=1,...,n-2$, then we have
\begin{equation}
N^{\beta}_{aA_1|A_2...A_n}\leq(1+{a})^{\frac{\beta}{s}-1}\sum_{i=1}^{n-1} \left((1+\frac{1}{a})^{\frac{\beta}{s}-1}\right)^{n-1-i}N_{a(i)}^{{\beta}}
\end{equation}
for $\beta\geq s$. 
\end{corollary}
\begin{example} Let us consider the three-qubit generlized $W$-class state,
\begin{eqnarray}\label{W}
|W\rangle_{A_1A_2A_3}=\frac{1}{2}(|100\rangle+|010\rangle)+\frac{\sqrt{2}}{2}|001\rangle.
\end{eqnarray}
 Then ${N_a}_{A_1|A_2A_3}=\frac{3}{4}$, ${N_a}_{A_1A_2}=\frac{1}{4},~{N_a}_{A_1A_3}=\frac{1}{2}$. 
 Let $a=2^{0.6}$, then the upper bound given in \cite{JFQ} is $$W_1=(\frac{1}{4})^\beta+\frac{(1+2^{0.6})^\frac{\beta}{s}-1}{(2^{0.6})^{\frac{\beta}{s}}}(\frac{1}{2})^\beta,$$
the upper bound given in \cite{ZJZ1,ZJZ2} is
 $$W_2=(\frac{1}{2})^{\frac{\beta}{s}}(\frac{1}{4})^\beta+\frac{(1+2^{0.6})^\frac{\beta}{s}-(\frac{1}{2})^{\frac{\beta}{s}}}{(2^{0.6})^{\frac{\beta}{s}}}(\frac{1}{2})^\beta,$$
 while our upper bound is $$W_3=(1+2^{0.6})^{\frac{\beta}{s}-1}(\frac{1}{4})^\beta+(1+2^{-0.6})^{\frac{\beta}{s}-1} (\frac{1}{2})^\beta.$$

 Fig. 2 charts our bound together with other bounds, and Fig. 3 and Fig. 4 show the comparison of our bound with those given in \cite{JFQ,ZJZ1,ZJZ2}. Our bound is found to be stronger than the other two.
\end{example}
\begin{figure}[H]
\centerline{\includegraphics[width=0.6\textwidth]{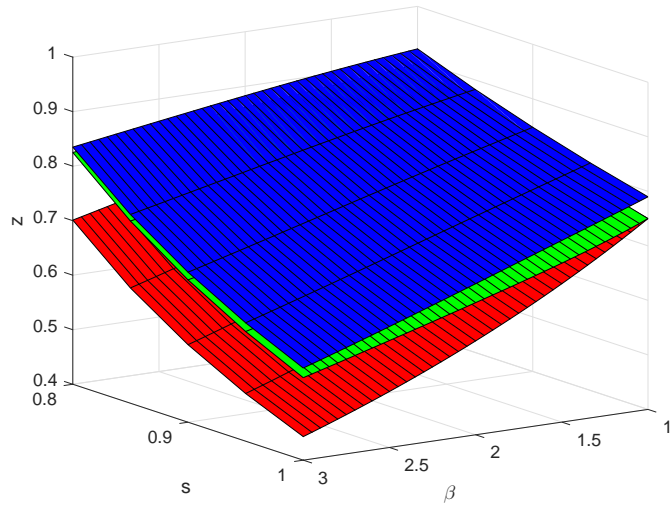}}
\renewcommand{\figurename}{Fig.}
\caption{The z-axis presents SCRENoA for the state $|W\rangle_{A_1A_2A_3}$ as a function of $\beta, s$. The red, green and blue surfaces chart
our upper bound, the upper bound of \cite{JFQ} and the upper bound of \cite{ZJZ1,ZJZ2} respectively.}
\end{figure}

\begin{figure}[H]
\centerline{\includegraphics[width=0.6\textwidth]{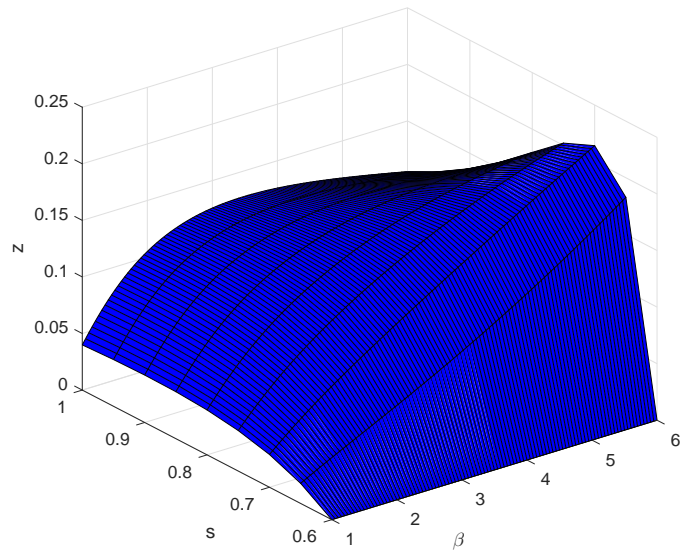}}
\renewcommand{\figurename}{Fig.}
\caption{The surface depicts our upper bound of SCRENoA minus that given by \cite{JFQ}.}
\end{figure}

\begin{figure}[H]
\centerline{\includegraphics[width=0.6\textwidth]{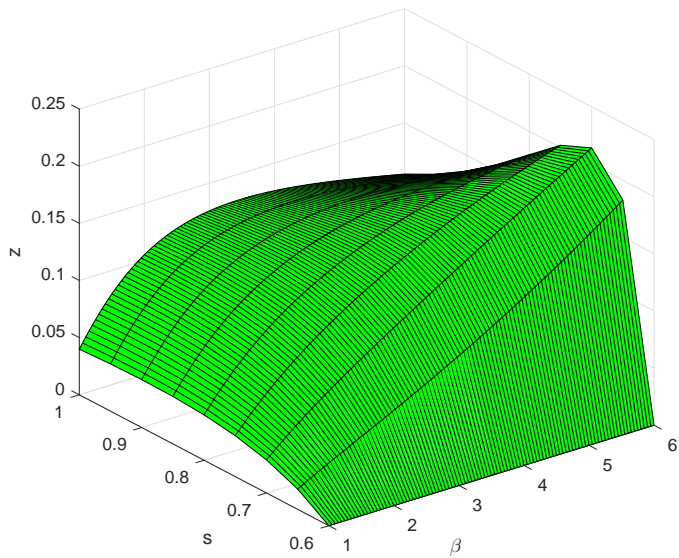}}
\renewcommand{\figurename}{Fig.}
\caption{The surface depicts our upper bound of SCRENoA minus that of \cite{ZJZ1,ZJZ2}.}
\end{figure}

\bigskip

\section{Conclusion}
Monogamy relations reveal special properties of correlations in terms of inequalities satisfied by various quantum measurements of the subsystems.
In this paper, we have examined the physical meanings and mathematical formulations related to monogamy and polygamy relations
in multipartite quantum systems. By grossly generalizing a technical inequality for the function $(1+t)^x$, we have obtained general stronger weighted
monogamy and polygamy relations for any quantum measurement such as concurrence, negativity, entanglement of formation etc. as well as
the Tsallis-$q$ entanglement and R\'enyi-$q$ entanglement measures. We have shown rigorously that our bounds outperform some of the strong bounds found recently
in a unified manner, notably that our results are not only stronger for monogamy relations but also polygamy relations. We have also used the concurrence and the SCRENoA
(square of convex-roof extended negativity of assistance) to show that our bounds are indeed better than the
recently available bounds through detailed examples and charts in both situations.

\bigskip

\noindent{\bf Acknowledgments}

This work is partially supported by Simons Foundation grant no. 523868
and National Natural Science Foundation of China grant no. 12126351.

\bigskip

\noindent\textbf {Data Availability Statements} All data generated or analysed during this study are included in the paper.

\bigskip
\bibliography{Bib}

\end{document}